\documentclass[a4paper,USenglish,cleveref,autoref]{lipics-v2019}

\title{Parameterization of tensor network contraction}

\author{
Bryan O'Gorman
}{
Berkeley Quantum Information \& Computation Center, University of California, Berkeley, CA, USA
Quantum Artificial Intelligence Laboratory, NASA Ames, Moffett Field, CA, USA
\and
}{
bogorman@berkeley.edu}{
https://orcid.org/0000-0001-5164-8083
}{}

\authorrunning{B. O'Gorman}

\Copyright{Bryan O'Gorman}



\funding{
The author was supported by a NASA Space Technology Research Fellowship.
}

\acknowledgements{
The author thanks Benjamin Villalonga for motivating this work, useful discussions, and feedback on the manuscript.
}

\nolinenumbers
\hideLIPIcs

\EventEditors{Wim van Dam and Laura Man\v{c}inska}
\EventNoEds{2}
\EventLongTitle{14th Conference on the Theory of Quantum Computation, Communication and Cryptography (TQC 2019)}
\EventShortTitle{TQC 2019}
\EventAcronym{TQC}
\EventYear{2019}
\EventDate{June 3--5, 2019}
\EventLocation{University of Maryland, College Park, Maryland, USA}
\EventLogo{}
\SeriesVolume{135}
\ArticleNo{7}

\usepackage{booktabs}
\usepackage{braket}
\usepackage{suffix}

\WithSuffix\newcommand\contractioncomplexity*{\mathrm{cc}}
\newcommand{\degree}{\mathrm{deg}}

\newcommand{\width}{\mathrm{width}}
\WithSuffix\newcommand\width*{\mathrm{w}}

\newcommand{\treewidth}{\mathrm{treewidth}}
\WithSuffix\newcommand\treewidth*{\mathrm{tw}}

\newcommand{\pathwidth}{\mathrm{pathwidth}}
\WithSuffix\newcommand\pathwidth*{\mathrm{pw}}

\newcommand{\branchwidth}{\mathrm{branchwidth}}
\WithSuffix\newcommand\branchwidth*{\mathrm{bw}}

\newcommand{\cutwidth}{\mathrm{cutwidth}}
\WithSuffix\newcommand\cutwidth*{\mathrm{cw}}

\newcommand{\modifiedcutwidth}{\mathrm{modcutwidth}}
\WithSuffix\newcommand\modifiedcutwidth*{\mathrm{mcw}}

\WithSuffix\newcommand\bramblenumber*{\mathrm{bn}}

\newcommand\eliminationwidth{\mathrm{elimwidth}}
\WithSuffix\newcommand\eliminationwidth*{\mathrm{ew}}

\newcommand{\congestion}{\mathrm{con}}
\newcommand{\vertexcongestion}{\mathrm{vertcon}}
\WithSuffix\newcommand\vertexcongestion*{\mathrm{vc}}
\newcommand{\edgecongestion}{\mathrm{edgecon}}
\WithSuffix\newcommand\edgecongestion*{\mathrm{ec}}
 
\begin{document}

\maketitle

\begin{abstract}
We present a conceptually clear and algorithmically useful framework for parameterizing the costs of tensor network contraction.
Our framework is completely general, applying to tensor networks with arbitrary bond dimensions, open legs, and hyperedges.
The fundamental objects of our framework are rooted and unrooted contraction trees, which represent classes of contraction orders.
Properties of a contraction tree correspond directly and precisely to the time and space costs of tensor network contraction.
The properties of rooted contraction trees give the costs of parallelized contraction algorithms.
We show how contraction trees relate to existing tree-like objects in the graph theory literature, bringing to bear a wide range of graph algorithms and tools to tensor network contraction.
Independent of tensor networks, we show that the edge congestion of a graph is almost equal to the branchwidth of its line graph.
\end{abstract}

\section{Introduction}

Tensor networks are widely used in chemistry and physics.
Their graphical structure provides an effective way for expressing and reasoning about quantum states and circuits.
As a model for quantum states, they have been very successful in expressing ansatzes in variational algorithms (e.g., PEPS, MPS, and MERA).
As a model for quantum circuits, they have been used in state-of-the-art simulations~\cite{villalonga2018flexible,dumitrescu2018benchmarking,dumitrescu2017tree,fried2018qtorch,pednault2017breaking}.
In the other direction, quantum circuits can also simulate tensor networks, in the sense that (additively approximate) tensor network contraction is complete for quantum computation~\cite{arad2010quantum}.

The essential computation in the application of tensor networks is tensor network contraction, i.e., computing the single tensor represented by a tensor network.
Tensor network contraction is $\#\textsc{P}$-hard in general~\cite{biamonte2015tensor} but fixed-parameter tractable.
Markov and Shi~\cite{markov2008simulating} defined the contraction complexity of a tensor network and showed that contraction can be done in time that scales exponentially only in the treewidth of the line graph of the tensor network.
Given a tree decomposition of the line graph of a tensor network, a contraction order can be found such that the contraction takes time exponential in the width of the decomposition, and vice versa.
However, the translation between contraction orders and tree decompositions does not account for polynomial prefactors.
This is acceptable in theory, where running times of $O(n2^n)$ and of $O(2^n)$ are both ``exponential'';
in practice, the difference between $\Theta(n2^n)$ and $\Theta(2^n)$ can be the difference between feasible and infeasible.

We give an alternative characterization of known results in terms of tree embeddings of the tensor network rather than tree decompositions of the line graph thereof.
In this context, we call such tree embeddings \emph{contraction trees}.
While one can efficiently interconvert between a contraction tree of a tensor network and a tree decomposition of the line graph, contraction trees exactly model the matrix multiplications done by a contraction algorithm in an abstract way.
That is, the time complexity of contraction is exactly and directly expressed as a property of contraction trees, in contrast to tree decompositions of line graphs, which only capture the exponent.
Our approach is thus more intuitive and precise,
and easily applies to tensor networks with arbitrary bond dimensions and open legs.

We show that contraction trees also capture the space needed by a matrix-multiplication-based contraction algorithm.
In practice, space often competes with time as the limiting constraint.
Even further, we can express the time used by parallel algorithms as a property of \emph{rooted contraction trees}, which are to contraction orders as partial orders are to total orders.

In a contraction tree, tensors are assigned to the leaves and each wire is ``routed'' through the tree from one leaf to another.
The congestion of a vertex of the contraction tree is the number of such routings that pass through it, and similarly for the congestion of an edge.
The vertex congestion of a graph $G$, denoted $\vertexcongestion*(G)$, is the minimum over contraction trees of the maximum congestion of a vertex, and similarly for the edge congestion, denoted $\edgecongestion*(G)$.
Formally, our main results are the following two theorems.

\begin{theorem}\label{thm:sequential-costs}
A tensor network $(G, \mathcal M)$ can be contracted in time 
$n 2^{\vertexcongestion*(G) + 1}$ and space 
$n 2^{\vertexcongestion*(G) + 1}$, or in time
$2^{1.5\edgecongestion*(G) + 1}$ and space 
$2^{\edgecongestion*(G) + 1}$.
More precisely, the tensor network can be contracted in time $\min_{(T, b)} \sum_{t \in T} 2^{\vertexcongestion*(t)}$, where the minimization is over contraction trees $(T, b)$.
The contraction can be done using space equal to the minimum weighted, directed modified cutwidth of a rooted contraction tree using edge weights $w(f) = 2^{\edgecongestion*(f)}$.
If the contraction is done as a series of matrix multiplications, these precise space and time bounds are tight (though not necessarily simultaneously achievable).
\end{theorem}

\begin{theorem}\label{thm:parallel-costs}
A parallel algorithm can contract a tensor network $(G, \mathcal M)$ in time 
\\
$\min_{(T, b)} \max_{l} \sum_{t} 2^{\vertexcongestion*(t)}$,
where the minimization is over rooted contraction trees $(T, b)$, the maximization is over leaves $l$ of $T$, and the summation is over vertices of $t$ on the unique path from the leaf $l$ to the root $r$.
In other words, the time is the minimum vertex-weighted height of a rooted contraction tree, where the weight of a vertex is $w(t) = 2^{\vertexcongestion*(t)}$.
If the contraction is done as matrix multiplications in parallel, this is tight.
\end{theorem}

Given a tree decomposition of a line graph with width $k-1$, we can efficiently construct a contraction tree of the original graph with vertex congestion $k$.
Thus one immediate application of our framework is as a way of precisely assessing the costs of contraction implied by different tree decompositions (even of the same width) computed using existing methods.
This is especially useful in distinguishing between contraction orders that have the same time requirements but different space requirements;
prior to this work, there was no comprehensive way of quantifying the space requirements, which in practice can be the limiting factor.
Alternatively, one can start with existing algorithms for computing good branch decompositions, which can be converted into contraction trees of small edge congestion.
More broadly, identifying the right abstraction (i.e., contraction trees) and precise quantification of the space and time costs is a foundation for minimizing those costs as much as possible.

In Section~\ref{sec:background}, we go over the graph-theoretic concepts that are the foundation of this work.
In Section~\ref{sec:unification}, we present seemingly unrelated graph properties in a unified framework that may be of independent interest.
Section~\ref{sec:unification}, while strictly unnecessary for understanding the main results, helps explain the relationship between our work and prior work.
In Section~\ref{sec:cost-model}, we introduce the cost model on which our results are based.
In Section~\ref{sec:contract-orders-and-trees}, we give our main results.
In Section~\ref{sec:ext-and-gen}, we discuss extensions and generalizations of the main ideas.
In Section~\ref{sec:conclusion}, we conclude with some possible directions for future work.
In Appendix~\ref{sec:branchwidth=edgecon}, we prove that the edge congestion of a graph is almost equal to the branchwidth of its line graph.

\begin{figure}
    \centering
    \begin{subfigure}{0.45\textwidth}
        \includegraphics[width=\textwidth]{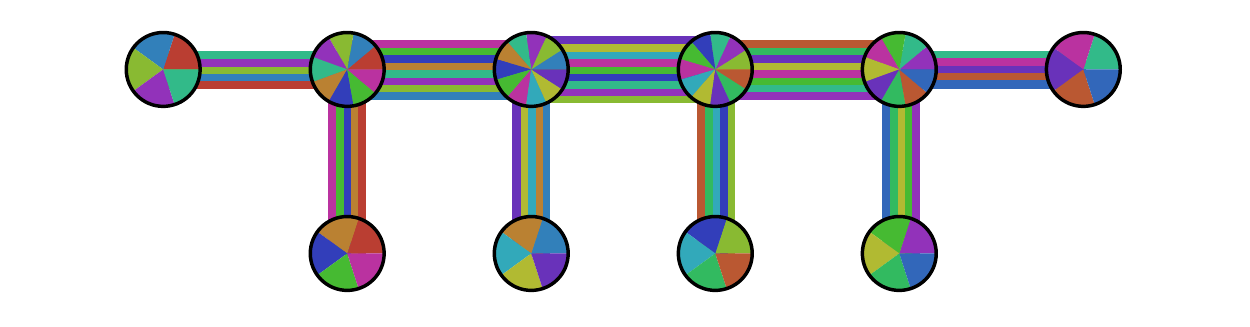}
    \end{subfigure}
    ~     \begin{subfigure}{0.45\textwidth}
        \includegraphics[width=\textwidth]{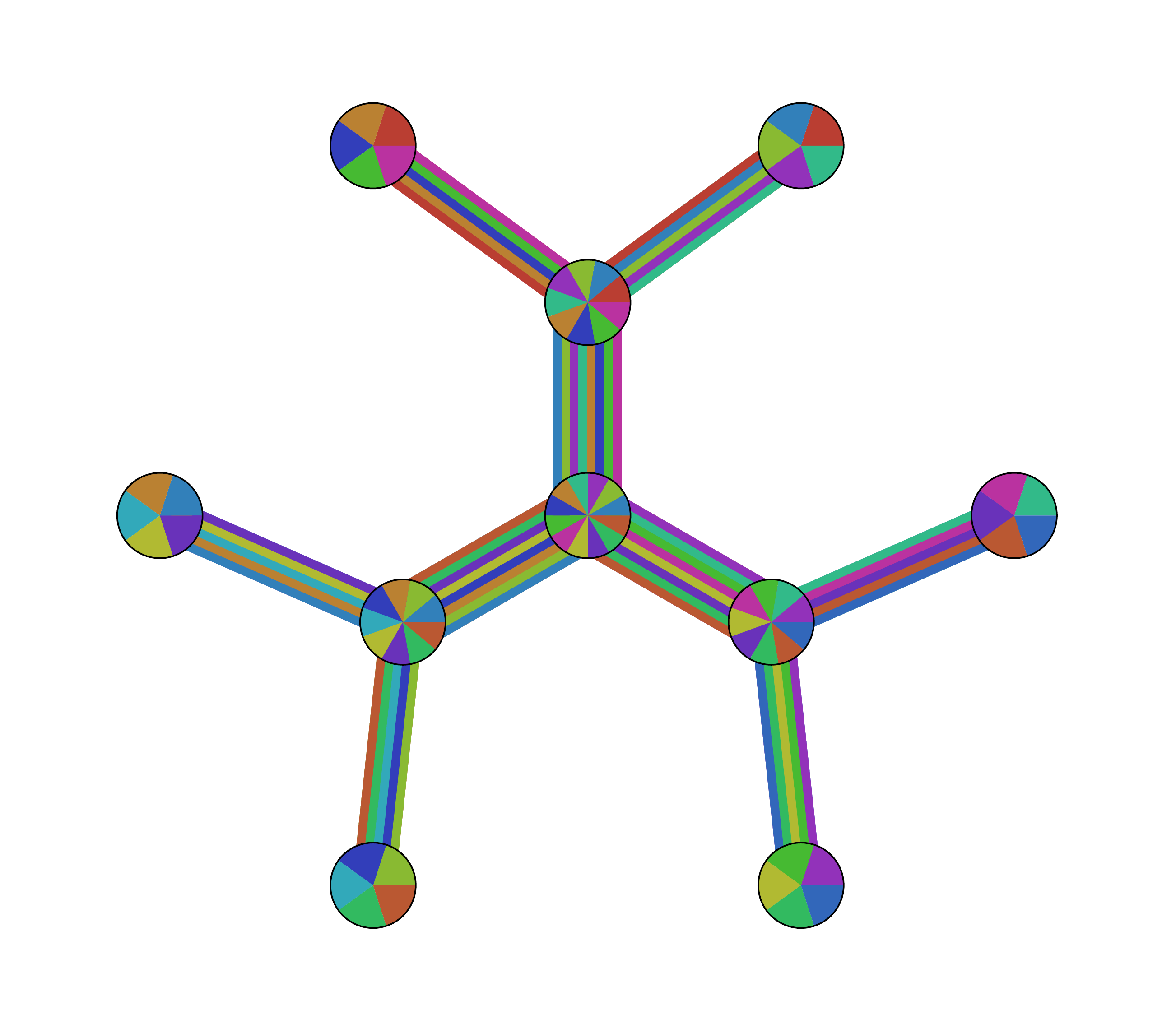}
    \end{subfigure}
\caption{Two unrooted contraction trees for a tensor network with $6$ tensors.
Each color corresponds to a wire of the tensor network.
The inclusion of a color in the representation of a vertex or edge of the contraction tree indicates the contribution of the corresponding wire's weight to the congestion of the vertex or edge, respectively.
}\label{fig:contraction trees}
\end{figure}
 \section{Background}\label{sec:background}

Let 
$[i, n] = \left\{j \in \mathbb Z \middle| i \leq j \leq n\right\}$,
$[n] = [1, n]$,
and
$[\mathbf n] = [n_1] \times [n_2] \times \cdots \times [n_r]$ for $\mathbf n = (n_1, \ldots, n_r) \in \mathbb {\left(Z^+\right)}^r$.
Let $G[S] = \left(V, E \cap \binom{S}{2}\right)$ be the subgraph of $G$ induced by a subset of the vertices $S \subset V(G)$.
For two disjoint sets of vertices of an edge-weighted graph,
$w(S, S') = \sum_{\left.\{u, v\} \in E \middle| u \in S, v \in S'\right.} w(\{u, v\})$
is the sum of the weights of the edges between $S \subset V$ and $S' \subset V$.
More generally, for $r$ disjoint sets of vertices,
$w(S_1, \ldots, S_r) = \sum_{\{i, j\} \in \binom{[r]}{2}} w(S_i, S_j)$ is the sum of the weights of the edges with endpoints in distinct sets.
In this context, we will denote singleton sets by their sole arguments, e.g.,
$w(u, v) = w (\{u\}, \{v\}) = w(\{u, v\})$.

\subsection{Tensor networks and contraction}
A \emph{tensor} can be defined in several equivalent ways. 
Most concretely, it is a multidimensional array.
Specifically, a rank-$r$ tensor is an $r$-dimensional array of size $\mathbf{d} = (d_1, \ldots, d_r)$.
More abstractly, a tensor is a collection of numbers indexed by the Cartesian product of some set of indices, e.g.,  
${\left[T_{i_1, i_2\ldots, i_r}\right]}_{(i_1, i_2, \ldots, i_r) \in [d_1] \times [d_2] \times \cdots \times [d_r]}$
indexed by $\mathbf i \in [\mathbf n]$.
Alternatively, a tensor can be thought of as a multilinear map $T: [\mathbf d] \to \mathbb C$.
(Our focus will be on complex-valued tensors.)

\begin{definition}[Tensor network]
A \emph{tensor network} $(G, \mathcal M)$ is an undirected graph $G$ with edge weights $w$ and a set of tensors $\mathcal M = \left\{\mathbf M_v \middle| v \in V(G)\right\}$ such that $\mathbf M_v$ is a $|N(v)|$-rank tensor with $2^{\degree(v)}$ entries, where $\deg(v) = \sum_{u \in N(v)} w(\{u, v\})$ is the weighted degree of $v$.
Each edge $e$ corresponds to an index $i \in \left[2^{w(e)}\right]$ along which the adjacent tensors are to be contracted.
\end{definition}

\begin{figure}
    \centering
    \begin{subfigure}{0.3\textwidth}
        \includegraphics[width=\textwidth, height=\textwidth]{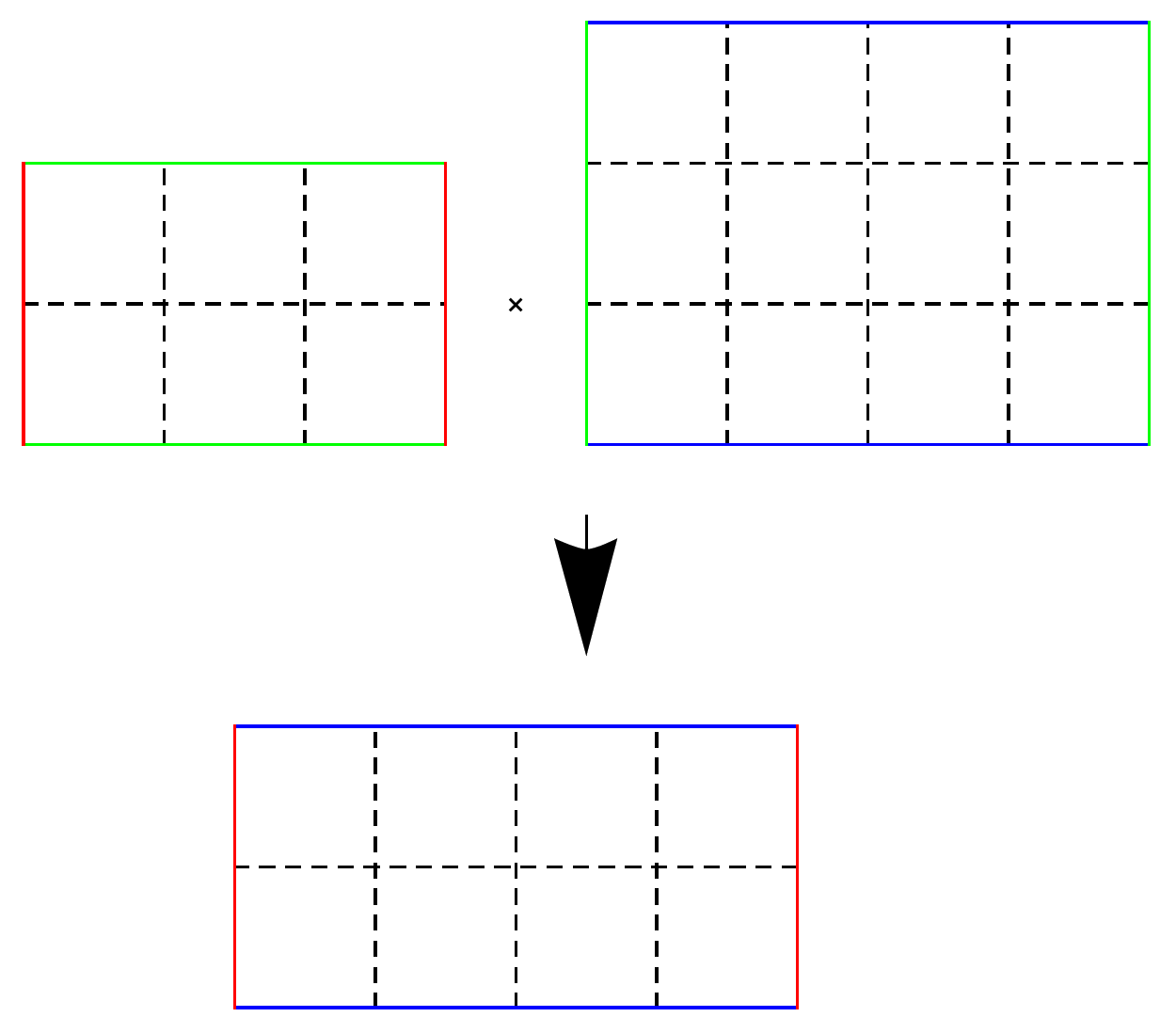}
    \end{subfigure}
    ~     \begin{subfigure}{0.3\textwidth}
        \includegraphics[width=\textwidth, height=\textwidth]{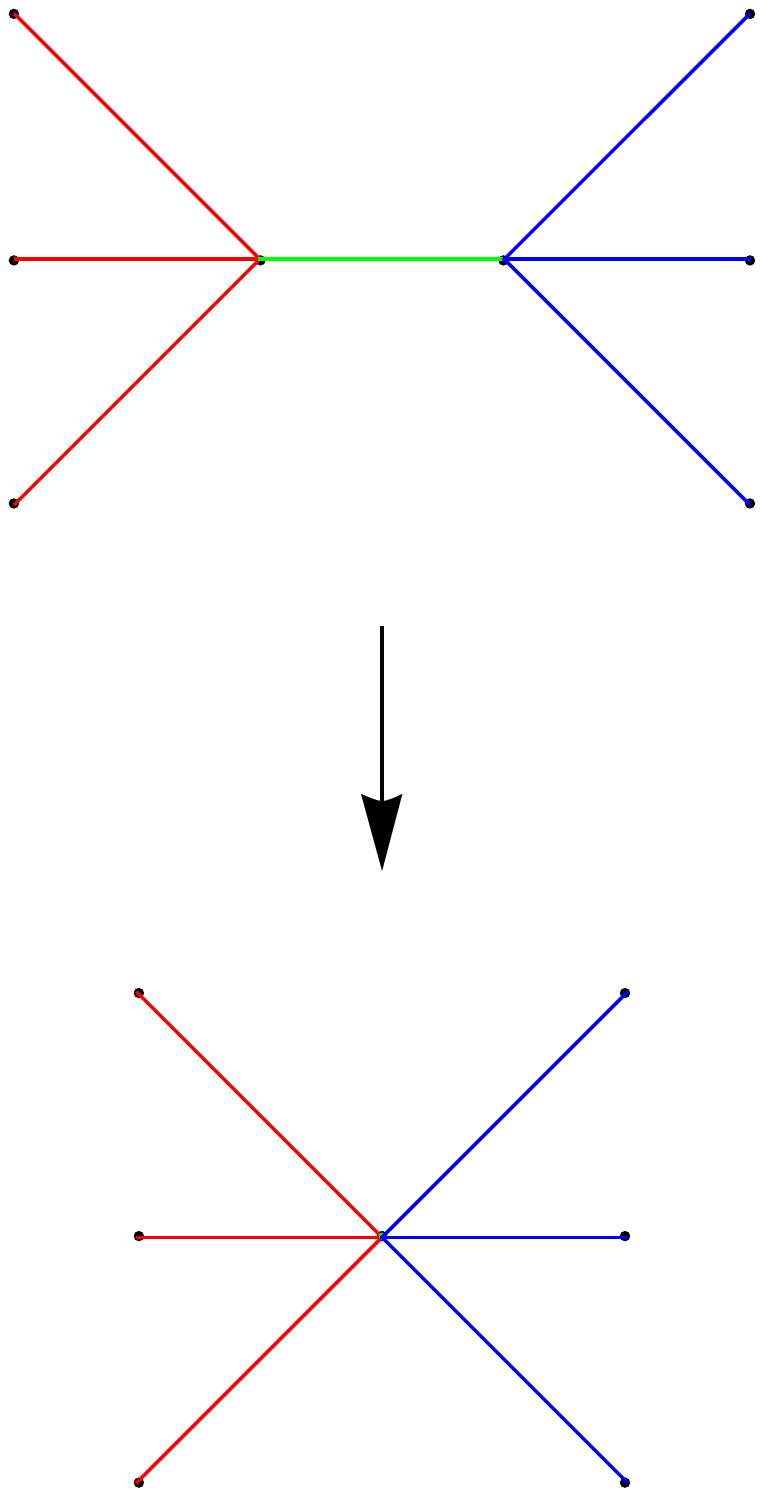}
    \end{subfigure}
    ~     \begin{subfigure}{0.3\textwidth}
        \includegraphics[width=\textwidth, height=\textwidth]{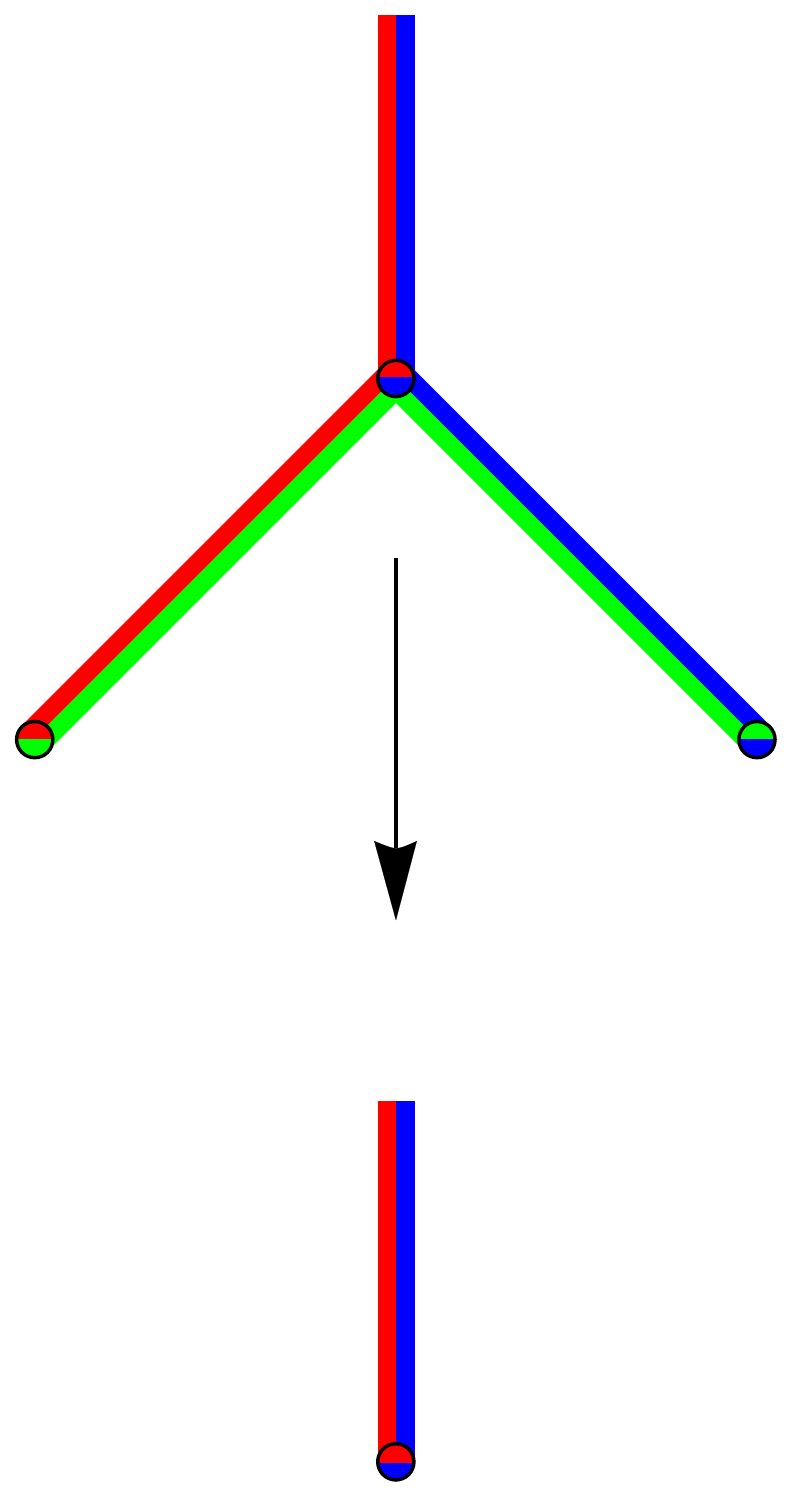}
    \end{subfigure}
\caption{Three ways of viewing the contraction of two tensors.
Left: multiplication of a $d_{\mathrm{L}} \times d_{\mathrm{M}}$-rank tensor with a $d_{\mathrm{M}} \times d_{\mathrm{R}}$-rank tensor, resulting in a $d_{\mathrm{L}}\times d_{\mathrm{R}}$-rank tensor.
Middle: contraction of a degree $w_{\mathrm{L}} + w_{\mathrm{M}}$ vertex with a degree $w_{\mathrm{M}} + w_{\mathrm{R}}$ vertex, resulting in a degree $w_{\mathrm{L}} + w_{\mathrm{R}}$ vertex.
Right: removing a close pair of leaves (with congestions $w_{\mathrm{L}} + w_{\mathrm{M}}$ and $w_{\mathrm{M}} + w_{\mathrm{R}}$) of a contraction tree, leaving a new leaf with congestion $w_{\mathrm{L}} + w_{\mathrm{R}}$.
}\label{fig:tensor-contraction}
\end{figure}

The contraction of two tensors is the summation over the values of their shared indices.
Graphically, this is like an edge contraction of the edge adjacent to the two tensors.
The result is a new tensor that takes the place of the two original ones.~\footnote{
Note that this is a ``parallel'' model of contraction, whereas Markov and Shi use ``one-edge-at-a-time'' contraction of multigraphs.
They are equivalent in the sense that an edge with integer weight can be considered as that number of (unweighted) parallel edges.
The parallel model more closely matches how contraction is done in practice.
It also allows for arbitrary bond dimension, whereas the multigraph model requires that all bond dimensions be powers of the same base.
}
See Figure~\ref{fig:tensor-contraction}.
Let $v_1$ and $v_2$ be the vertices contracted into the new vertex $v_{\{1,2\}}$.
The weight of an edge between the new vertex and any other vertex $v'$ is 
$w\left(v_{\{1, 2\}}, v'\right) = 
w\left(\{v_1, v_2\}, v'\right) = 
w\left(v_1, v'\right) + w\left(v_2, v'\right)$.

Except in Section~\ref{sec:ext-and-gen}, we assume that all tensor networks have no ``open legs'', i.e., every edge connects two vertices (tensors).
In this case, the value of a tensor network is the single number that results from contracting all of its edges.
Each contraction reduces the number of vertices (tensors) by one, so the network is fully contracted by $n-1$ contractions.
We call a sequence of such contractions a \emph{contraction order}.
The value of the tensor network is independent of the contraction order, but the cost of doing the contraction can vary widely depending on the contraction order.
Each contraction is identified by an edge, but that edge may not be in the original graph, i.e., its adjacent vertices may have been formed by earlier contractions.
One way of specifying a contraction order is by a sequence of edges of the original graph that constitute a spanning tree thereof.
In Section~\ref{sec:contract-orders-and-trees}, we introduce the notion of contraction trees, which allow for a conceptually clear way of expressing contraction orders that makes manifest the associated temporal and spatial costs. 

Exactly computing the value of a tensor network is \#P-hard~\cite{biamonte2017tensor}, as a tensor network can be constructed that counts the number of satisfying assignments to a satisfiability instance or the number of proper colorings of a graph.
Even multiplicative and additive approximation is NP-hard~\cite{arad2010quantum}.
Interestingly, approximating the value of a tensor network with bounded degree and bounded bond dimension is BQP-complete~\cite{arad2010quantum}.
That is, not only can tensor networks simulate quantum circuits, but quantum circuits can simulate tensor networks as well.
In this sense, tensor networks and quantum circuits are computationally equivalent.

\subsection{Treewidth and branchwidth}
This section is intended primarily to establish notation and recapitulate the standard definitions of the graph properties used in the present work.
For a more thorough and pedagogical treatment, see Diestel's excellent textbook~\cite{diestel2018graph}.
Many instances of graph problems that are hard in general are actually easy when instance graphs are restricted to trees.
In many such cases, this generalizes in the sense that it is possible to characterize the hardness of an instance by how ``tree-like'' it is, as captured by the treewidth of the graph.
The treewidth of a graph is defined in terms of an optimal tree decomposition.
Treewidth has several alternative characterizations; one of these, \emph{elimination width}, is the basis of Markov and Shi's result equating treewidth and contraction complexity.

\begin{definition}[Tree decomposition]
A \emph{tree decomposition} of a graph $G=(V, E)$ is a tuple $(T, \mathcal X)$ of a tree $T$ and a tuple $\mathcal X = {\left(X_t\right)}_{t \in V(T)}$ of subsets (called \emph{bags}) of the vertices of $G$ with the following properties.
\begin{enumerate}
    \item For every edge $\{u, v\} \in E(G)$, there is some bag $X \in \mathcal X$ that contains both endpoints: $u, v \in X$.
    \item For every vertex $v \in V(G)$ of $G$, the subtree $T[S_v]$ of $T$ induced by the bags 
    $S_v = \left\{X \in \mathcal X \middle| v \in X\right\}$ containing $v$ is non-empty and connected.
\end{enumerate}
\end{definition}

\begin{definition}[Width and treewidth]
The \emph{width} of a tree decomposition $(T, \mathcal X)$ of a graph $G$ is one less than the size of the largest bag: $\width(G, T, \mathcal X) = \width(\mathcal X) = \max_{X \in \mathcal X} |X| - 1$.
The \emph{treewidth} of a graph is the minimum width of a tree decomposition of the graph.
\end{definition}
\noindent
A related concept is that of path decompositions and pathwidth, defined analogously to tree decompositions and treewidth, except restricted to paths rather than trees.

\begin{definition}[Path decomposition and pathwidth]
A \emph{path decomposition} of a graph $G$ is a tree decomposition $(T, \mathcal X)$ of $G$ such that $T$ is a path.
The pathwidth $\pathwidth(G)$ of $G$ is the minimum width of a path decomposition of $G$.
\end{definition}

\begin{definition}[Branch decomposition]
A \emph{branch decomposition} of a graph $G = (V, E)$ is a tuple $(T, b)$ of a binary tree $T$ and a bijective function $b: E(G) \to V(T)$ between the edges $E$ of $G$ and the leaves of $T$.
\end{definition}

For each vertex $v \in V(G)$ of $G$, let $S_v \subset V(T)$ be the minimal spanning tree of $T$ that contains all the leaves corresponding to edges adjacent to $v$.

\begin{definition}[Branchwidth]
The \emph{width}, denoted $\width_{G}(T, b, \{s, t\})$, of an edge $\{s, t\} \in E(T)$ of a branch decomposition $(T, b)$ of a graph $G$ is $\left|\left\{ v \in V(G) \middle| \{s, t\} \subset S_v\right\}\right|$, i.e., the number of vertices of $G$ such that the subtree $T[S_v]$ contains $\{e, t\}$.
The \emph{width} of the branch decomposition is the largest width of an edge, $\width_G(T, b) = \max_{f \in E(T)} \width_G(T, b, f)$.
The \emph{branchwidth} $\branchwidth(G) = \min_{(T, b)} \width_G(T, b)$ of a graph is the minimum width of a branch decomposition thereof.
\end{definition}

\subsection{Congestion}

There is an alternative but less explored way of quantifying how ``tree-like'' a graph is: the minimum congestion of a tree embedding, introduced by Bienstock~\cite{bienstock1990embedding}.\footnote{
Note that this is entirely distinct from a different type of congestion problem in which the goal is find routings for some specified set of pairs of terminals.}

\begin{definition}[Tree embedding]
A \emph{tree embedding} of a graph $G$ is a tuple $(T, b)$ of a binary tree $T$ and a bijection $b: V(G) \to V(T)$ between the vertices of $G$ and the leaves of $T$.
\end{definition}

\noindent 
Let $S_{v, w}$ be the unique path between the leaves $b(v)$ and $b(w)$ of $T$.

\begin{definition}[Congestion]
The congestion of a vertex $v\in V(T)$ (resp., edge $f\in E(T)$) is the total weight of the edges $e \in E(G)$ whose subtrees $S_e$ include $v$ (resp., $f$).
\end{definition}

\subsection{Cutwidth}

\begin{definition}[Cutwidth]
Let $f: V \to [n]$ be a linear ordering of the vertices of a graph $G = (V, E)$. 
The cutwidth of $f$ is the maximum number of edges that cross a gap: 
\begin{equation*}
\max_{i \in [n-1]} \left|\left\{\{u, v\} \in E \middle| f(u) \leq i < f(v)\right\}\right|.
\end{equation*}
The modified cutwidth of $f$ is the maximum number of edges that cross a vertex:
\begin{equation*}
\max_{i \in [n]} \left|\left\{\{u, v\} \in E \middle| f(u) < i < f(v)\right\}\right|.
\end{equation*}
The cutwidth (resp., modified cutwidth) of a graph is the minimum cutwidth (resp., modified cutwidth) of a linear ordering.
For edge weighted graphs, the weighted cutwidth and modified cutwidth count the total weights of the relevant edge sets rather than their cardinalities.
For a directed acyclic graph, the directed cutwidth (resp., modified cutwidth) is the minimum cutwidth (resp., modified cutwidth) of a linear ordering that is topologically sorted according to the graph.
\end{definition}

\subsection{Parameterized complexity}

Approximating both treewidth and pathwidth to within a constant factor is NP-hard, though there exist efficient algorithms for logarithmic and polylogarithmic approximations, respectively~\cite{bodlaender1992approximating,bodlaender1995approximating}.
However, deciding whether or not the treewidth is at most some constant can be done in linear time (albeit it with an enormous prefactor)~\cite{bodlaender1993linear}.
For many graph problems, e.g., Maximum Independent Set, there exist algorithms whose run time is exponential only in the treewidth or pathwidth, i.e., given the instance graph and a tree decomposition thereof of width $k$, the algorithm runs in time $2^k n^{O(1)}$~\cite{arnborg1989linear}.
The Exponential Time Hypothesis (ETH) implies that several such parameterized complexity results are optimal, in the sense that there exists no $2^{o(k)} n^{O(1)}$ algorithm~\cite{cygan2015lower}.

The situation is similar for branchwidth.
Computing the branchwidth of a graph is in general NP-hard, but can be done efficiently for planar graphs~\cite{seymour1994call}.
(Whether computing the treewidth of a planar graph is NP-hard is an open question.)
As is the case for treewidth, there is a constructive linear time algorithm for deciding whether or not the branchwidth is at most some constant (and in this case with better constant factors)~\cite{bodlaender1997constructive}.
Good branch decompositions can be used to implement dynamic programming algorithms for problems such as the traveling salesman problem~\cite{cook2003tour}.

Computing the vertex congestion of a graph is claimed to be NP-hard~\cite{bienstock1990embedding}, but no proof appears in the literature.

Computing the (edge) cutwidth is NP-hard, but for any constant $k$, a linear ordering of cutwidth $k$ (for all variants) can be found in linear time if one exists~\cite{bodlaender2009derivation}.
 \section{Unified framework of graph properties}\label{sec:unification}

\begin{figure}
\centering
\begin{tabular}{c|c|c||c|c} &&& \multicolumn{2}{c}{Target family} \\
Leaves & Subtrees & Minimization over & Trees & Caterpillars \\ \bottomrule
Edges & Vertices & Vertices & Treewidth & Pathwidth \\
Edges & Vertices & Edges & Branchwidth & \\
Vertices & Edges & Vertices & Vertex congestion & Modified cutwidth \\
Vertices & Edges & Edges & Edge congestion &
\end{tabular}
\caption{
Table of graph properties.
Each row corresponds to an instantiation of Equation~\ref{eq:unified-definition}.
}
\label{fig:graph-properties}
\end{figure}

In this section, we present a unified framework of various graph properties, as captured in the following combined definition:

\begin{equation}\label{eq:unified-definition}
\parbox{\textwidth}{
A 
$\left\{\substack{
    \text{tree decomposition} \\ 
    \text{branch decomposition} \\ 
    \text{tree embedding} \\
    \text{tree embedding}
}\right\}$
of a 
$\left\{\substack{
    \text{vertex} \\
    \text{vertex} \\ 
    \text{edge} \\ 
    \text{edge}
}\right\}$-weighted
graph $G$ is a tuple $(T, b)$ of a binary tree $T$ and a bijection $b$ between the leaves of $T$ and the  $\left\{\substack{
    \text{edges} \\ 
    \text{edges} \\ 
    \text{vertices} \\
    \text{vertices}
}\right\}$
of $G$.
The bijection $b$ implies a subtree for every 
$\left\{\substack{
    \text{vertices} \\
    \text{vertices} \\ 
    \text{edges} \\ 
    \text{edges}
}\right\}$
of the graph.
The 
$\left\{\substack{
    \text{treewidth} \\
    \text{branchwidth} \\
    \text{vertex congestion} \\
    \text{edge congestion} \\
}\right\}$
of the graph is the minimum over all
$\left\{\substack{
    \text{tree decompositions} \\ 
    \text{branch decompositions} \\ 
    \text{tree embeddings} \\
    \text{tree embeddings}
}\right\}$
of the maximum total weight of all subtrees containing any 
$\left\{\substack{
    \text{vertex} \\
    \text{edge} \\ 
    \text{vertex} \\ 
    \text{edge}
}\right\}$.
The 
$\left\{\substack{
    \text{pathwidth}\\
    \text{---}\\
    \text{modified cutwidth}\\
    \text{---}
}\right\}$
is defined in the same way as 
$\left\{\substack{
    \text{treewidth} \\
    \text{---} \\
    \text{vertex congestion} \\
    \text{---} \\
}\right\}$
except that $T$ is restricted to be a caterpillar.
}
\end{equation}

Let's unpack this.
For branchwidth and congestions,~\eqref{eq:unified-definition} is the standard definition.
For the others,~\eqref{eq:unified-definition} is non-standard but equivalent to the standard definitions.
Writing them all in this way helps elucidate the relationships between them, which are obscured by the standard definitions.

Note that both the vertex and edge congestions of a graph $G$ are defined as optimal properties of the same type of object, namely a tree embedding $(T, b)$.
For every edge $e \in E(G)$, the mapping $b: V(G) \to V(T)$ of the vertices to leaves of the tree implies a minimal subtree $S_e$ connecting the leaves of $T$ corresponding to its endpoints in $G$. 
(For an edge of size $2$, this subtree $S_e$ is a path, but the definition allows for hyperedges as well.)
The vertex and edge congestions are then the maximum total weight of subtrees that contain any vertex or edge, respectively, of the tree $T$.

There is a similar relationship between treewidth and branchwidth.
Usually, we think of a tree decomposition of a graph $G = (V, E)$ as a tree $T$ and a subtree $S_v$ for every vertex in $V(G)$ such that the subtrees for every pair of adjacent (in $G$) vertices overlap.
In~\eqref{eq:unified-definition}, $S_v$ is specified implicitly as the (unique) minimal spanning subtree of $T$ that connects the leaves of $T$ corresponding to the edges of $G$ that are incident to $v$.
By design, this tree $T$ and set of subtrees is the same as that for a branch decomposition.
The treewidth and branchwidth are the maximum total weight of subtrees (now corresponding to vertices of $G$) that contain any vertex or edge, respectively, of the tree $T$.

So we see that the congestions are defined by the overlap of subtrees of $T$ corresponding to edges of $G$ and that the tree- and branchwidths are defined by the overlap of subtrees of $T$ corresponding to vertices of $G$, the former implied by a mapping from vertices of $G$ to leaves of $T$ and the latter by a mapping from edges of $G$ to leaves of $T$.
The vertex congestion and treewidth are concerned with the overlap at vertices of $T$, and the edge congestion and branchwidth with the overlap on edges of $T$.
Thus we have made the analogy that
treewidth : branchwidth :: (vertex congestion) : (edge congestion). For example, that~\cite{robertson1991obstructions}
$\branchwidth*(G) \leq \treewidth*(G) \leq \frac32 \branchwidth*(G)$
and~\cite{bienstock1990embedding}
$
\edgecongestion*(G) \leq \vertexcongestion*(G) \leq \frac32 \edgecongestion*(G)$
is no coincidence.

Now consider the line graph 
$L(G) = \left(E, \left\{\{e, e'\} \subset E \middle| e \cap e' \neq \emptyset\right\}\right)$ of a graph $G = (V, E)$.
Suppose we have a tree embedding $(T, b)$ of the original graph $G$, with an implied subtree $T_e$ for every edge $e \in E(G)$.
Because the vertices of the line graph $L(G)$ correspond to the edges of $G$, this can be considered as a branch decomposition of the line graph $L(G)$.
For every pair of edges $e, e' \in E(G) = V(L(G))$ that are adjacent in the line graph, the corresponding subtrees $S_e, S_{e'}$ intersect at the leaf $b(v) \in V(T)$, where $v \in e, e'$ is the vertex of $G$ adjacent to $e$ and $e'$.
The vertex congestion of the tree embedding $(T, b)$ is the width of $(T, b)$ interpreted as a tree decomposition, and 
the edge congestion of the tree embedding is the width of $(T, b)$ interpreted as a branch decomposition.
This implies that $\treewidth*(L(G)) \leq \vertexcongestion*(G)$ and 
$\branchwidth*(L(G)) \leq \edgecongestion*(G)$.
Actually, these inequalities are tight or almost so: 
$\treewidth*(L(G)) = \vertexcongestion*(G)$ and 
$\branchwidth*(L(G)) \leq \edgecongestion*(G) \leq 
\branchwidth*(L(G)) + \left\lfloor \frac{\degree(G)}{3} \right\rfloor$.
The other direction, going from a tree decomposition to a tree embedding, requires seeing that a tree decomposition of a line graph can be made to have a particular structure, specifically that the edges of $L(G)$ corresponding to each vertex of $G$ can be mapped to disjoint subtrees of $T$.
The equality was shown by Harvey and Wood~\cite{harvey2018treewidth} and captures how our characterization of the temporal costs of tensor network contraction relates to earlier characterizations.
However, our characterization in terms of tree embeddings, while mathematically equivalent to that in terms of tree decompositions of line graphs, allows for a conceptually cleaner and more fine-grained perspective.  
We prove the inequalities in Appendix~\ref{sec:branchwidth=edgecon}.
\begin{theorem}
\label{thm:branchwidth-equals-edge-congestion}
The edge congestion of graph $G$ is at least the branchwidth of its line graph and at most the same plus a third of its maximum degree.
Furthermore, a tree embedding with edge congestion $k + \left\lfloor\degree(G) / 3\right\rfloor$ can be efficiently computed from a branch decomposition of width $k$ and a branch decomposition of width $k$ can be efficiently computed from a tree embedding with edge congestion $k$.
\end{theorem}

For vertex congestion and treewidth (which concern the overlap of subtrees at vertices), the requirement that the mapping be a bijection with the leaves of the tree can be dropped, as can the requirement that the tree be binary.
Yet these requirements are without loss of generality, as any tree embedding or tree decomposition can be modified to satisfy these without increasing its vertex congestion or treewidth, respectively.
For edge congestion and branchwidth, which concern overlap over edges, the bijection and degree requirements are essential.

The usual definitions for pathwidth and modified cutwidth are in terms of paths (or, equivalently, linear orderings), whereas in~\eqref{eq:unified-definition} we allowed them to be caterpillars.
This is equivalent, and allows us to relate the properties just discussed with their linear variants.
In particular, the relationship between the bubblewidth of a tensor network $(G, \mathcal M)$ and its ``contraction complexity'' is almost the same as that between the modified cutwidth of the graph and its vertex congestion, in the sense that the bubblewidth is exactly equal to the cutwidth and $\cutwidth*(G) \leq \modifiedcutwidth*(G) \leq \cutwidth*(G) + \degree(G)$. 

We can make another analogy, that
treewidth : (vertex congestion) :: pathwidth : (modified cutwidth) .
For example,~\cite{bodlaender2009derivation,harvey2018treewidth}
$\frac12 \left(\treewidth*(G) + 1\right)
\leq \vertexcongestion*(G) + 1 \leq \degree(G)
\left(\treewidth*(G) + 1\right)
$
and
$
{\pathwidth*(G) \leq \modifiedcutwidth*(G) + 1 \leq \degree(G)
\left(\pathwidth*(G) + 1\right)}
$.

The (unmodified) cutwidth is a linear analog to what Ostrovskii called the ``tree congestion'' of a graph~\cite{ostrovskii2004minimal}; the tree congestion is the same as the edge congestion except that there is a bijection between \emph{all} the vertices of the binary tree, rather than just the leaves.
 \section{Contraction costs}\label{sec:cost-model}

Our primary motivation is minimizing the time and space costs of tensor network contraction.
Ideally, for instances of interest we would like to \emph{provably} minimize the cost, which entails tight lower bounds and the corresponding constructions that meet them.
Given the formal hardness of tensor network contraction and the informal hardness of proving lower bounds, we restrict our attention to minimizing the cost of tensor network contraction as it is most commonly done: as a series of matrix multiplications.

First, how much space is required to store a tensor network $(G, \mathcal M)$?
Each tensor $\mathbf M_v$ consists of $2^{\degree_v}$ numbers; this is the main component of the space requirements.
Technically, we must also keep track of the graph $G$ and the weights of its edges $E(G)$ as well as a dope vector for each tensor indicating how the tensor is laid out in memory; these will be negligible. 
Our memory accounting will be in units of whatever is used to store a single entry of a tensor.
While in general, the bit depth of an entry may scale non-trivially with instance size, practical implementations will use a fixed-width data type. 

Then, what do we need to do a contraction of two tensors?
Suppose we want to contract a $(d_{\mathrm{L}}, d_{\mathrm{M}})$ tensor with a $(d_{\mathrm{M}}, d_{\mathrm{R}})$ tensor along their shared dimension $d_{\mathrm{M}}$.
The input tensors require a total of $d_{\mathrm{M}} (d_{\mathrm{L}} + d_{\mathrm{R}})$ space and the output tensor $d_{\mathrm{L}} d_{\mathrm{R}}$. 
In theory, it should be possible to do the contraction using no more space than that required by the larger of the input tensors and output tensor.
In practice, new memory is allocated for the new tensor, it is populated with the appropriate data from the input tensors, and then the memory for the latter is freed. 
We assume the second cost model, in which memory is simultaneously allocated both for the tensors to be contracted and for the tensor that results from their contraction, but our ideas are straightforwardly modifiable for plausible variants.

The contraction itself is essentially matrix multiplication, and a straightforward implementation will take time $d_{\mathrm{L}} \cdot d_{\mathrm{M}} \cdot d_{\mathrm{R}}$.
There exist Strassen-like algorithms for matrix multiplication with better asymptotic runtime, but the constant pre-factors are so large and the straightforward algorithm so heavily optimized that they are of little practical value given the size of currently available machines.

Lastly, in order to implement a tensor contraction as a matrix multiplication, the tensors must be laid out commensurately in memory.  
If they are not, then the data of one tensor or both must be permuted to make them so. 
This permutation can effectively be done in place and in linear time.
In practice, the permutation time is negligible  compared to the matrix multiplication time.
 \section{Contraction orders and trees}\label{sec:contract-orders-and-trees}

In this section, we present our main contribution: a graph-theoretic characterization of the temporal and spatial costs of families of contraction orders.

\subsection{Linear contraction orders}

\begin{figure}[t!]
\centering
\includegraphics[width=0.7\textwidth]{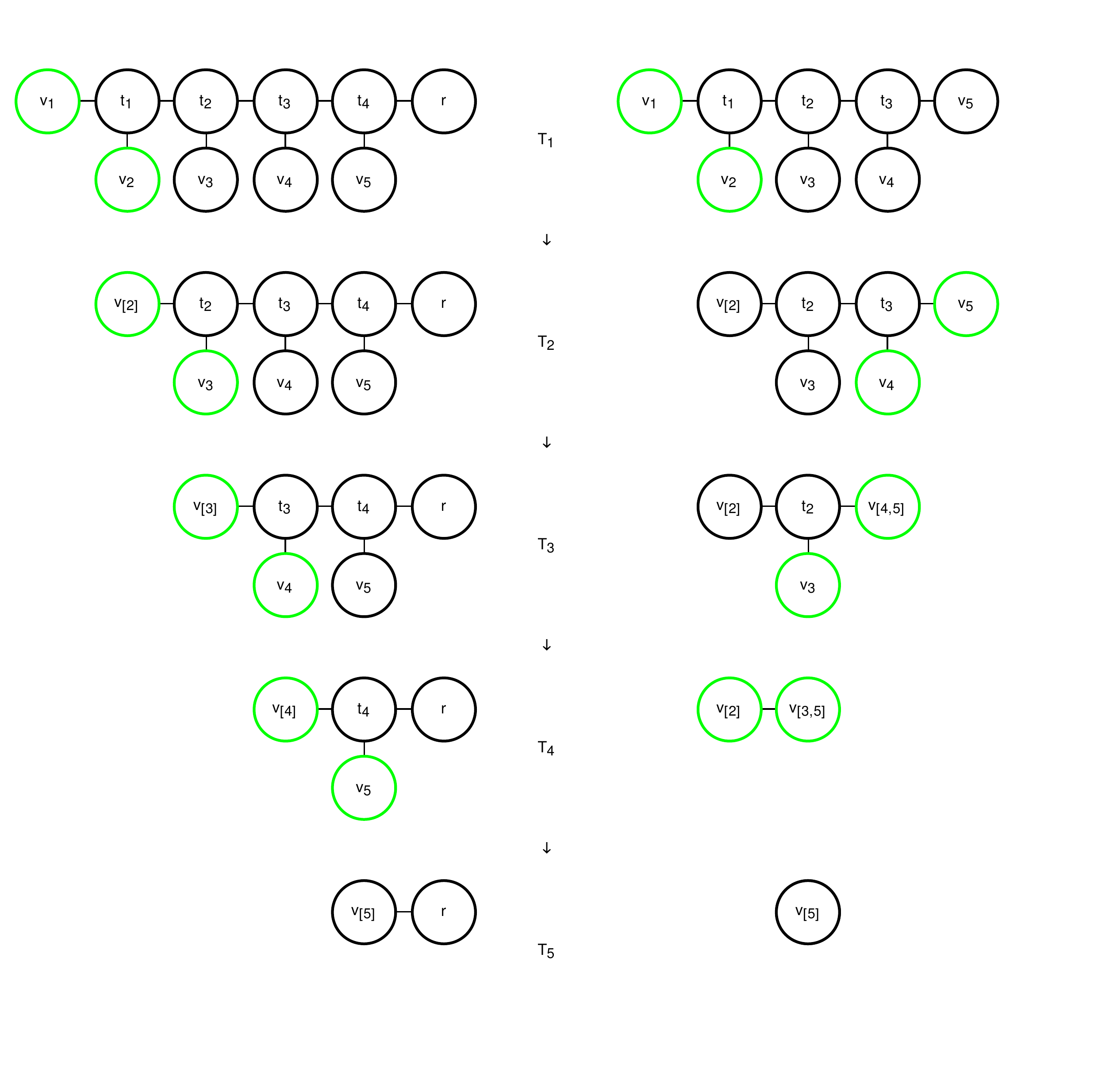}
\caption{Two series of contraction trees (unrooted and rooted on left and right, respectively) for a tensor network with $5$ tensors.
From top to bottom, the contraction trees for the initial, intermediate, and final tensor networks.
The pair of leaves corresponding to the next pair of tensors to be contracted are highlighted in green.
}\label{fig:linear-contraction-trees}
\end{figure}

\begin{figure}
\begin{center}
\includegraphics[width=0.15\textwidth,angle=90]{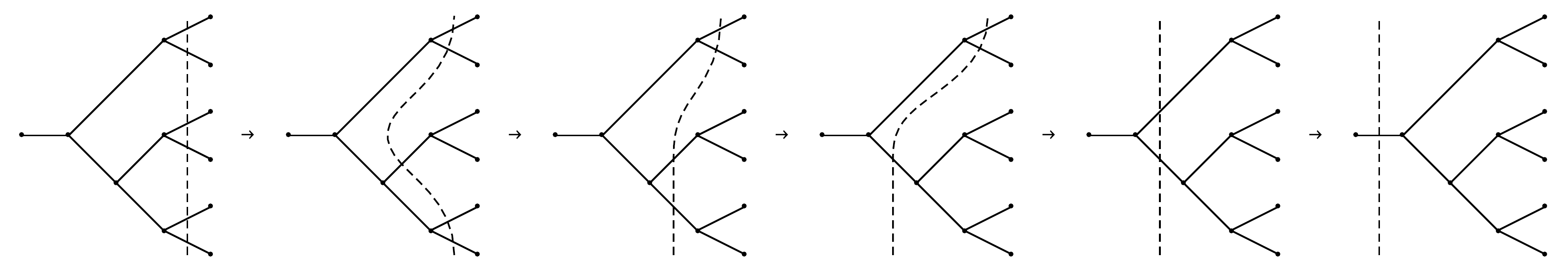}
\caption{The intermediate states of a contraction procedure.
The tree pictured is a rooted contraction tree, with the root at the left.
The dashed line crosses edges of the contraction tree adjacent to tensors held in memory.
}\label{fig:frontiers}
\end{center}
\end{figure}
We start with a special case of contraction orders.
Let a \emph{linear contraction order} be one specified by an ordering of the vertices 
$(v_1, v_2, \ldots, v_n)$.
That is, the first contraction is of vertices $v_1$ and $v_2$ to form a new vertex $v_{1,2}$.
The second contraction is of $v_{1,2}$ and $v_3$ to form $v_{1,2,3}$, and so on.
We represent such a contraction order by what we call a \emph{rooted contraction tree}.
The contraction tree of a linear contraction order is a binary caterpillar tree with $n+1$ leaves, one for each vertex of the original graph and a special leaf called the \emph{root}, as shown in Figure~\ref{fig:linear-contraction-trees}.
The root leaf is at one of the ``ends'' of the tree.
Each vertex $v_i$ for $i \in [2, n]$ is at distance $n + 2 - i$ from the root, and vertex $v_1$ is at distance $n$ therefrom.
We denote such a contraction tree by $(T, b)$, where $T$ is the tree and $b: V(G) \to V(T) \cup \{r\}$ is the bijection between the vertices of $G$ and the leaves of $T$ together with the root $r$.

Recall that for a tensor network $(G, \mathcal M)$, we are using the convention that the weight $w(u, v)$ of an edge $\{u, v\}$ is the logarithm of the bond dimension of wire connecting tensors $\mathbf M_u$ and $\mathbf M_v$.
For each edge $\{u, v\}$ of $G$ there is a unique path in $T$ between $b(u)$ and $b(v)$, which we call a \emph{routing}. 
Assign the weight $w(u, v)$ to every vertex and edge on this path, including the endpoints $b(u)$ and $b(v)$.
We say that the \emph{congestion} of a vertex or edge of $T$, denoted $\congestion(v)$ or $\congestion(e)$, is the sum of the weights of all the routings that include it.
Label the non-root leaves of $T$ by $l_i = b(v_i)$ and the internal vertices by $t_i$ for $i \in [n - 1]$, where $t_{n-1}$ is closest to the root and $t_1$ is farthest.
For concision, identify $t_0$ with $t_1$.

We now show that these congestions capture the costs of the contraction order.
First, note that for each vertex $v \in G$, the congestion $\congestion(e)$ of the edge $e \in E(T)$ adjacent to $b(v)$ gives the size of the tensor $\mathbf M_v$, in the sense that $\congestion(e) = \sum_{u \in V(G)} w(v, u) = \degree_v$, so that 
$2^{\congestion(e)}$ is the product of the bond dimensions of the tensor $M_v$.
Now, consider the first contraction, of vertices $v_1$ and $v_2$, i.e., tensors $\mathbf M_{v_1}$ and $\mathbf M_{v_2}$.
The bond dimension of the wire between them is $2^{w(v_1, v_2)}$.
The product of the bond dimensions of $\mathbf M_{v_1}$ with tensors besides $v_2$ is $2^{\deg_{v_1} - w(u, v)}$, and similarly for $\mathbf M_{v_2}$.
As discussed in Section~\ref{sec:cost-model}, the contraction can be done in time 
$2^{\deg_{v_1} - w(u, v)} \cdot 2^{w(u,v)} \cdot 2^{\deg_{v_2} - w(v_1, v_2)} =
2^{w(v_1, v_2, V(G) \setminus \{v_1, v_2\})}$, where $w(v_1, v_2, V(G) \setminus \{v_1, v_2\})$ is the total weight of edges across the tripartite cut.
This is exactly the congestion of the vertex $t \in V(T)$ adjacent to both $b(v_1)$ and $b(v_2)$.
Suppose that we have done the contraction, yielding a new tensor network containing the contracted vertex $v_{1, 2}$.
The size of this new tensor $\mathbf M_{v_{1,2}}$ is $2^{w(\{v_1, v_2\}, V \setminus \{v_1, v_2\})} = 2^{\congestion(\{t_1, t_2\})}$.
If we continue with the contractions, we notice an exciting pattern.
We can identify each contraction with an internal vertex of $T$.
The congestion of that vertex gives the time to do the contraction, and the congestion of the adjacent edge nearest the root gives the space of the resulting contracted tensor.
The congestion of the leaves, which is equal to the congestions of the adjacent edges and gives the size of the corresponding tensors, can be interpreted as giving the time required to simply read in the tensors of the initial network to be contracted.
Overall, the total time of all the contractions is 
$\sum_{t \in V(T)} 2^{\congestion(t)} \leq 2n \cdot 2^{\vertexcongestion_{T, b}(G)}$,
where  $\vertexcongestion_{T, b}(G) = \max_{t \in V(T)} \congestion(t)$.
Furthermore, each edge $e \in E(T)$ corresponds to a tensor $\mathbf M_e$ that appears at some point in the series of contractions; those adjacent to leaves correspond to the initial tensors and internal edges to tensors resulting from contractions.
The congestion of each edge gives the size of the corresponding tensor, in the sense that the size of $\mathbf M_f$ is $2^{\congestion(f)}$.
At any point point in the contraction order, there are at most $n$ tensors, so the required memory is at most $n 2^{\edgecongestion_{T, b}(G)}$, where
$\edgecongestion_{T, b}(G) = \max_{f \in E(T)} \congestion(f)$.
As shown in Section~\ref{sec:unification}, the minimum vertex congestion over all linear contraction orders is exactly equal the vertex cutwidth of $G$.
It is closely related to the bubblewidth of earlier work~\cite{arad2010quantum}, which is exactly equal to the edge cutwidth.
The minimum edge congestion over all linear contraction orders is exactly equal to the edge cutwidth of $G$.

\subsection{General contraction orders}

We now turn our attention to general (i.e., not necessarily linear) contraction orders.
The first generalization we make is to remove the root.
In other words, for each linear contraction order we form an \emph{unrooted contraction tree} exactly as before except that leaves of $T$ are in unqualified bijection with the vertices of $G$.
This unrooted contraction tree can be interpreted as corresponding to $2^{n-2}$ different contraction orders in the following way.
Define a pair of leaves in a binary tree to be \emph{close} if they are at distance $2$.
In the caterpillar binary trees we have considered thus far, there are two pairs of close leaves, at each ``end'' of the tree.
Before, we used a rooted caterpillar contraction tree to represent the unique contraction order given by contracting the two non-root close leaves until we got to the root.
Now, the unrooted caterpillar contraction tree represents the family of contraction orders that can be specified by contracting \emph{either} pair of a close leaves of the contraction tree until a single vertex remains.
Importantly, it remains true that every one of these contraction orders takes time exactly 
$\sum_{t \in V(T)} 2^{\congestion(t)} = 
\tilde{O}\left(2^{\vertexcongestion_{T, b}(G)}\right)$
and space 
$\tilde{O}\left(2^{\edgecongestion_{T, b}(G)}\right)$.

The second generalization we make is to remove the restriction to caterpillar trees.

\begin{definition}
A \emph{rooted contraction tree} $(T, b)$ of a tensor network $(G, \mathcal M)$ is a rooted binary tree $T$ and a bijection $b: V(G) \to V(T)$ between the vertices (tensors) of $G$ and the (non-root) leaves of $T$.
An \emph{unrooted contraction tree} $(T, b)$ is an unrooted binary tree $T$ and a bijection $b: V(G) \to V(T)$ between the vertices of $G$ and the leaves of $T$.
\end{definition}

\noindent
An unrooted contraction tree represents a set of contraction orders in the following way.
Suppose we have a contraction order $e_1, \ldots, e_{n-1}$;
Each edge can be written as $e_i = \{v_{S}, v_{S'}\}$ for some disjoint $S, S' \subset V(G)$, where $v_{S}$ is the vertex formed by contracting the vertices in $S$.
We start with an empty forest $T_1 = (V(G), \emptyset)$.
For each contraction $e_i = \{v_{S}, v_{S'}\}$, we add a new vertex $v_{S \cup S'}$ to the forest, as well as edges from the new vertex to $v_S$ and $v_{S'}$.
That is, 
$T_i = 
\left(
    V(T_{i-1}) \cup \{v_{S \cup S'}\},
    E(T_{i-1}) \cup \left\{\{v_{S}, v_{S \cup S'}\}, \{v_{S'}, v_{S \cup S'}\}\right\}
\right)$.
For the last contraction, instead of adding a new vertex, we only add an edge between $v_{S}$ and $v_{S'}$.
Doing this yields an unrooted contraction tree for the given contraction order.
We say that an unrooted contraction tree represents the set of contraction orders from which it can be constructed in this way.
If for the last contraction, we added not only a new vertex $v_V$ connected to $v_{S}$ and $v_{S'}$ but a second vertex $r$ connected $v_V$, we would have a rooted contraction tree.

We can easily go the other way as well.
Suppose we have a rooted contraction tree.
Then we can iteratively build a contraction order.
We select an arbitrary close pair of non-root leaves, and add to the contraction order the contraction corresponding to the adjacent internal vertex (of the tree).
We then remove the two leaves and their adjacent edges.
The internal vertex now becomes a leaf, and corresponds to the vertex resulting from contracting the two vertices (of the tensor network) in the new contraction tree.
We repeat until only a single edge of the tree remains, corresponding to the completely contracted tensor network and the root.
This is the same procedure visualized in Figure~\ref{fig:linear-contraction-trees}, except that, when the contraction tree is not restricted to be a caterpillar, there may be many more than two pairs of close leaves to choose from at each step.

Given an unrooted contraction tree, we can turn it into a rooted contraction tree by splitting any edge (i.e., removing an edge, adding a new vertex and adding edges between the new vertex and the vertices adjacent to the removed one), and then adding a root vertex and connecting it to the first newly inserted vertex.

\begin{proof}[Proof of Theorem~\ref{thm:sequential-costs}]
In a contraction tree, either rooted or unrooted, each internal vertex corresponds to a contraction.
In rooted contraction trees, there is a clear directionality; two of the neighbors are ``inputs'' and the third is ``output''.
However, the congestion of the vertex, the exponential of which gives the time to do the matrix multiplication, is independent of this directionality.
Similarly, each edge of a contraction tree corresponds to a tensor that exists at some point in the contraction (specifically, when the edge is adjacent to a leaf).
Again the congestion of this edge is independent of its direction, and the size of the tensor is the exponential of the congestion.
Without loss of generality, we prove the theorem using rooted contraction trees.

Suppose we have a rooted contraction tree $(T, b)$ of a tensor network $(G, \mathcal M)$.
Each internal vertex $i \in V(T)$ corresponds to a matrix multiplication, which takes time $2^{\vertexcongestion*(i)}$.
Each leaf $l \in V(T)$ corresponds to an initial tensor of size $2^{\vertexcongestion*(l)}$, where 
$\vertexcongestion*(l) = \degree_G\left(b^{-1}(l)\right)$.
Overall, the total time then is 
$\sum_{t \in V(T)} 2^{\vertexcongestion*(t)} \leq 2n 2^{\vertexcongestion*(T)}$.

The rooted contraction tree gives a partial ordering of its vertices, which represent contractions (or initial tensors).
Any topologically sorted linear ordering $(t_1, t_2, \ldots, t_{2n-2})$ of the vertices of the contraction tree can be considered uniquely as a contraction order consistent with the contraction tree, and vice versa.
For a given contraction order, consider the intermediate state at some point in the overall contraction procedure.
Let $t_i$ be the last tensor contracted and $t_{i+1}$ the next one to contract.
Each edge $f \in E(T)$ from $\{t_1, \ldots, t_i\}$ to $\{t_{i+1}, \ldots, t_{2n-2}\}$ corresponds to a tensor that needs to be stored at this point.
The size of the tensor is exactly $2^{\edgecongestion*(f)}$.
The size of the next tensor (resulting from the contraction corresponding to $t_{i+1}$) is $2^{\edgecongestion*(f')}$, where $f'$ is the edge from $t_{i+1}$ towards the root of $T$.
Using the convention that the weight of an edge $f \in E(T)$ of $T$ is $w(f)=2^{\edgecongestion*(f)}$, then the directed, weighted modified cutwidth of a vertex $t_{i+1}$ in a linear ordering $(t_1, \ldots, t_{2n-2})$ of the vertices of $T$ is exactly equal to the space needed to store the remaining tensors to be contracted and make room for the tensor resulting from the next contraction.
Once the contraction is done, the memory allocated for the two tensors that were contracted can be freed.
For the coarser space bound, we can just pre-allocate memory for every tensor that will arise during the procedure, in total space 
$\sum_{f \in E(T)} 2^{\edgecongestion*(f)} \leq 2n 2^{\edgecongestion*(T)}$.

Overall, if we choose a contraction tree $T$ with minimum vertex congestion, i.e.,  $\vertexcongestion*(G) = \vertexcongestion*(T) \geq \edgecongestion*(T)$, we get 
time at most $n 2^{\vertexcongestion*(G)}$ and space at most $n 2^{\vertexcongestion*(G)}$.
If instead we choose a contraction tree $T$ with minimum edge congestion, i.e., $\edgecongestion*(G) = \edgecongestion*(T) \geq (2/3) \vertexcongestion*(T)$, we get time at most $n 2^{1.5\edgecongestion*(G) + 1}$ and space at most $n 2^{\edgecongestion*(G)}$.
Tightness follows from the fact that for any contraction order, we can construct a rooted contraction tree whose properties give the stated bounds.
\end{proof}

\begin{proof}[Proof of Theorem~\ref{thm:parallel-costs}]
Suppose we have a rooted contraction tree $(T, b)$ and that $l^* \in V(T)$ is a leaf on a longest path from a leaf to the root using the vertex weight 
$w(t) = 2^{\vertexcongestion*(t)}$.
Call this path from $l^*$ to the root the \emph{critical path} $P*$.
The vertices on $P*$, ordered from the leaf to the root, represent a series of contractions.
This series of contractions can be done in time $\sum_{t \in V(P^*)} 2^{\vertexcongestion*(t)}$, the vertex-weighted length of the $P^*$, which by definition is the longest such path.
We prove the claim for general contraction trees by induction.
The base case is a tensor network of just two tensors, so that there is just a single contraction and the critical path has $3$ vertices.
The inductive step is that if the claim is true for a contraction tree whose critical path has $k$ vertices, it is true for a contraction tree whose critical path has $k+1$ vertices. 
Consider the last vertex $t$ on $P^*$ nearest the root.
It corresponds to a contraction of a tensor from an earlier contraction $P*$ and a tensor from the remaining subtree of $T$, i.e., the part of tree not containing $P^*$.
By definition, the length of the critical path of this subtree is no more than the length of the subpath from $l^*$ to $t$; otherwise $P^*$ would not be the longest path.
Therefore, this subtree can be contracted in less time than the earlier parts of $P^*$.
These can be done in parallel, so the overall time is simply that for $P^*$.
\end{proof}

As shown in Appendix~\ref{sec:branchwidth=edgecon}, a branch decomposition of $L(G)$ with width $k$ can be efficiently converted into a contraction tree of $G$ with edge congestion $k + \left\lfloor\degree(G) / 3\right\rfloor$.
Similarly, a tree decomposition of $L(G)$ with width $k-1$ can be efficiently converted into a contraction tree of $G$ with vertex congestion $k$~\cite{harvey2018treewidth}.
Thus one way of utilizing these results is to use an existing algorithm for finding tree decompositions or branch decompositions as a starting point.
Theorems~\ref{thm:sequential-costs} and~\ref{thm:parallel-costs} can then be used to construct minimum-cost contraction orders in a more precise way than previous results allow.
Developing empirically good implementations of algorithms for finding tree decompositions is a particularly active area research~\cite{pace2017treewidth}.
These are already exploited in much recent work on tensor network contraction~\cite{boixo2017simulation,chen2018classical,villalonga2018flexible}.
The framework presented here can significantly augment the effectiveness of such techniques.
For instances with a lot of structure, as typical ones do, the intuitiveness of contraction trees also empowers manual construction of contraction trees.

There are also techniques for certifying the optimality of tree decompositions and branch decompositions (namely, brambles and tangles) that can be ported to certify the optimality of contraction trees with respect to vertex and edge congestion, respectively.
For planar graphs, the exact edge congestion can be computed (non-constructively) in polynomial time~\cite{seymour1994call}.
In addition to serving as a lower bound for calculations, the structure of such obstructions may help with understanding the complexity of quantum states as represented by tensor networks.

 \section{Extensions and generalizations}\label{sec:ext-and-gen}

\begin{figure}
\begin{center}
\includegraphics[width=0.3\textwidth]{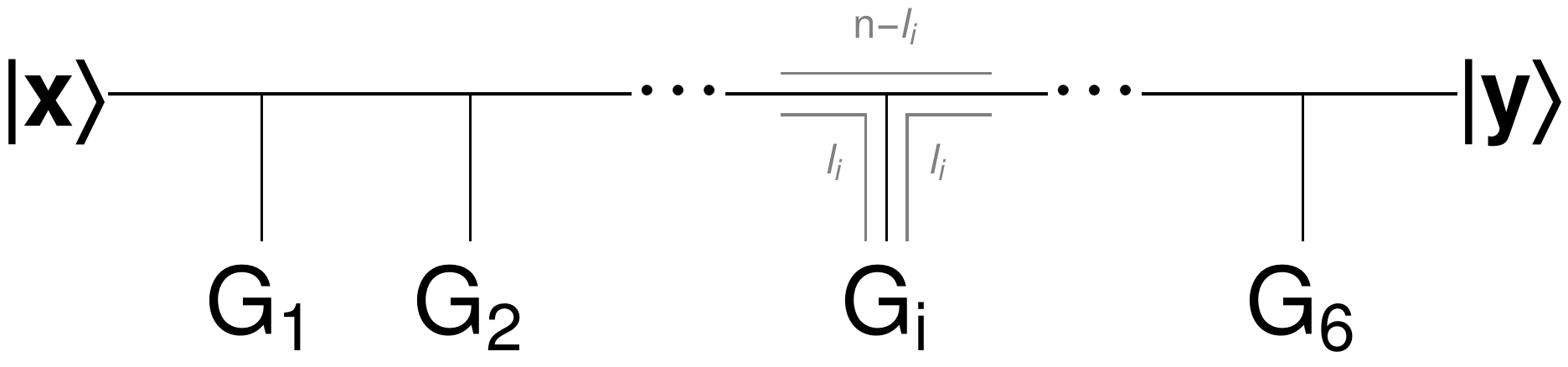}
\end{center}
\caption{Contraction tree representation of the Schr\"{o}dinger algorithm for computing $\braket{x | G_1 G_2 \cdots G_m | y}$.
}\label{fig:schroedinger}
\end{figure}

Heretofore, we have assumed that all tensor networks under consideration had no open legs, i.e., that they contract to a single number (0-rank tensor).
More generally, we can consider tensor networks with open legs that contract to non-trivial tensors.
For such tensors, we treat any open legs as wires to a single ``environment'' tensor, which we then identify with the root of a rooted contraction tree.
For the purposes of minimizing the congestion, the graph will simply have one more vertex.
All previous results regarding the costs of contraction then follow exactly as before without modification.

We can also allow tensor networks $(G, \mathcal M)$ in which $G$ is a hypergraph.
Recall how we defined the congestions of a contraction tree $(T, b)$.
Each vertex $v\in V(G)$ was identified with a leaf of $T$ through the bijection $b$.
Then each edge $\{u, v\} \in E(G)$ contributed its weight to the congestions of the vertices and edges on the routing (unique path) between $b(u)$ and $b(v)$ in $T$.
For a hyperedge $\{v_1, \ldots, v_k\}$, there is a unique subtree of $T$ connecting the adjacent vertices (which is equal to the union of the paths connecting each pair of edges).
Then the hyperedge contributes its weight to the congestions of the vertices and edges on this subtree.
The hyperedge corresponds to a so-called ``copy'' tensor with $k$ legs of the same bond dimension $b$~\cite{biamonte2017tensor}.
The copy tensor is one when all indices have the same value and is zero otherwise.
Such a tensor arises, e.g., in a decomposition of a controlled quantum gate.

Decompositions of tensors highlight the main limitation of the present work.
While our upper bounds are unconditional, our lower bounds hold only within what we call the matrix multiplication model, in which the only operations allowed are matrix multiplications.
This takes advantage only of the topological properties of $G$ and, importantly, not of the properties of the tensor $\mathcal M$. 
However, in many cases of practical interest, the tensors have structure that can be exploited.
For example, a tensor corresponding to a quantum gate can be split into two tensors connected by a wire with bond dimension equal to the Schmidt rank across some bipartition of the qubits on which it acts.
For gates with less-than-full Schmidt rank, this can help with contraction significantly.
Once such a decomposition is made, the sparser graph structure can be exploited by the methods presented here.

Tree-based methods for tensor network contraction are used in state-of-the-art simulations of quantum circuits, where ``simulation'' here means calculation a single matrix element $\braket{\mathbf{x} | C | \mathbf{y}}$ for a pair of basis states $(\ket{\mathbf{x}}, \ket{\mathbf{y}})$ and the circuit $C$.
In addition to providing a precise analysis of such methods, we can also analyze algorithms not usually expressed in such terms.
For example, consider the ``Schr\"{o}dinger'' algorithm: a state vector of size $2^n$ is kept in memory and for each of $m$ gates in sequence.
Suppose each gate acts on at most $l$ qubits.
Let the circuit be represented as a tensor network with $m + 2$ tensors: one for each gate, one for the output $\ket{\mathbf{x}}$, and one for the input $\ket{\mathbf{y}}$.
In the corresponding graph $G$, the vertex $\ket{\mathbf{x}}$ is adjacent to each of the gate vertices that first act on a qubit, with weight equal to the number of qubits that are first acted on by the gate.
Similarly for $\ket{\mathbf{y}}$.
The Schr\"{o}dinger algorithm is then a linear contraction order using the vertex ordering $\left(\ket{\mathbf{x}}, G_1, \ldots, G_m, \ket{\mathbf{y}}\right)$.
Each internal vertex of the contraction tree adjacent to the vertex corresponding to the $l_i$-local gate $G_i$ has congestion $n+l_i$: $n - l_i$ from the qubits not acted on by the gate, then $l_i$ each from the input and output wires. 
The total time for the contraction is thus $\sum_{i=1}^m 2^{n+l_i} \leq m 2^{n+l} = O(m2^n)$, where $l = O(1)$ is the maximum locality of a gate.
Each internal edge has congestion $n$, so the contraction can be done using space $O(2^n)$.

An alternative approach is the ``Feynman'', or path integral, algorithm, which inserts resolutions of the identity after every gate and sums.
Now we consider the tensor network corresponding to $\braket{\mathbf{x} | C | \mathbf{y}}$ slightly differently.
For simplicity, assume all gates are $2$-local.
Instead of having a single vertex $\ket{x}$ for the input, we have $n$ vertices $\ket{x_i}$, one for each qubit.
Similarly, we have $n$ output vertices $\left\{\ket{y_i}\right\}$.
First, we contract the input vertices $\ket{x_i}$ into the adjacent gate vertices.
This leaves $2m$ wires, $2$ from each gate to the next or an output.
Suppose that instead of contracting the entire tensor network, we remove a single wire and replace it with $\ket{b}\bra{b}$ for $b\in\{0, 1\}$.
The value of the original network is the sum of the values of the reduced networks over $b \in \{0, 1\}$.
The Feynman algorithm is then to do this for all wires.
For each value $\mathbf{b} \in {\{0, 1\}}^{2m}$, we have a tensor network of $m$ tensors and no wires, which we can ``contract'' in $O(m)$ time and $O(1)$ space.
But we need to do this for every $\mathbf{b}$ and sum them up, meaning overall it takes $O\left(m4^m\right)$ time.
We need $O(n + m)$ space to keep track of $\mathbf{x}$, $\mathbf{y}$, and $\mathbf{b}$.
We can generalize this approach to arbitrary tensor networks.
First, we remove some set $S \subset E(G)$ of edges, with total weight $W=\sum_{s \in S} w(e)$.
There are $2^W$ values of the corresponding wires, and for each one we contract the reduced tensor network.
Let $\tilde{G} = (V, E \setminus S)$ be the reduced network.
Overall, for a sequential algorithm, this takes time $\tilde{O}\left(2^{W + \vertexcongestion*(\tilde{G})}\right)$ and space $\tilde{O}\left(W + m 2^{\edgecongestion*(\tilde{G})}\right)$.
Moreover, we consider the cuts $S$ as allowing trivial parallelization, by doing the $2^W$ contractions of the reduced network in parallel on the same number of processors.
This idea was used, for example, by Villalonga et al.\ to balance time and memory usage in their simulation of grid-based random quantum circuits.
Aaronson and Chen~\cite{aaronson2016complexity} show that for carefully chosen cuts that form nested partitions, the contributions $W$ to the time and space from the cuts can be significantly reduced.
 \section{Conclusion}\label{sec:conclusion}
We introduced a graph-theoretic framework for precisely quantifying the temporal and spatial costs of tensor network contraction, with the ultimate goal of minimizing these costs.
We conclude with several possible directions for future work:
\begin{itemize}
\item Proving the hardness of exactly or approximately computing the vertex or edge congestion of a graph, including of special cases like planar graphs. 
\item Inventing algorithms (that aren't simply disguised treewidth or branchwidth algorithms) for finding small-congestion contraction trees.
\item 
Exploring the space-time trade-off of vertex and edge congestions. 
They are always within a small multiplicative constant  of each other, but can they be exactly minimized simultaneously?
If not, what does the trade-off look like, particularly for graphs of practical interest like 2D and 3D grids.
\item Parallelizing at larger scale.
In our discussion of parallelized algorithms, we neglected communication costs. While this is probably reasonable at a relatively small number of parallel processes (i.e., that can be on a single multi-processor node of a cluster), at larger scales it may become material and worth trying to minimize.
\item Adapting our methods to approximate tensor network contraction.
\item Finding analogous methods for optimizing tensor-network ansatzes.
For example, it is known that optimizing (bounded-bond dimension) tree tensor networks is easy.
Can this be generalized in a parameterizable way as we did for tensor network contraction?
\end{itemize}
 
\appendix
\section{Branchwidth and edge congestion}\label{sec:branchwidth=edgecon}

\begin{figure}[t!]
\centering
\includegraphics[width=0.9\textwidth]{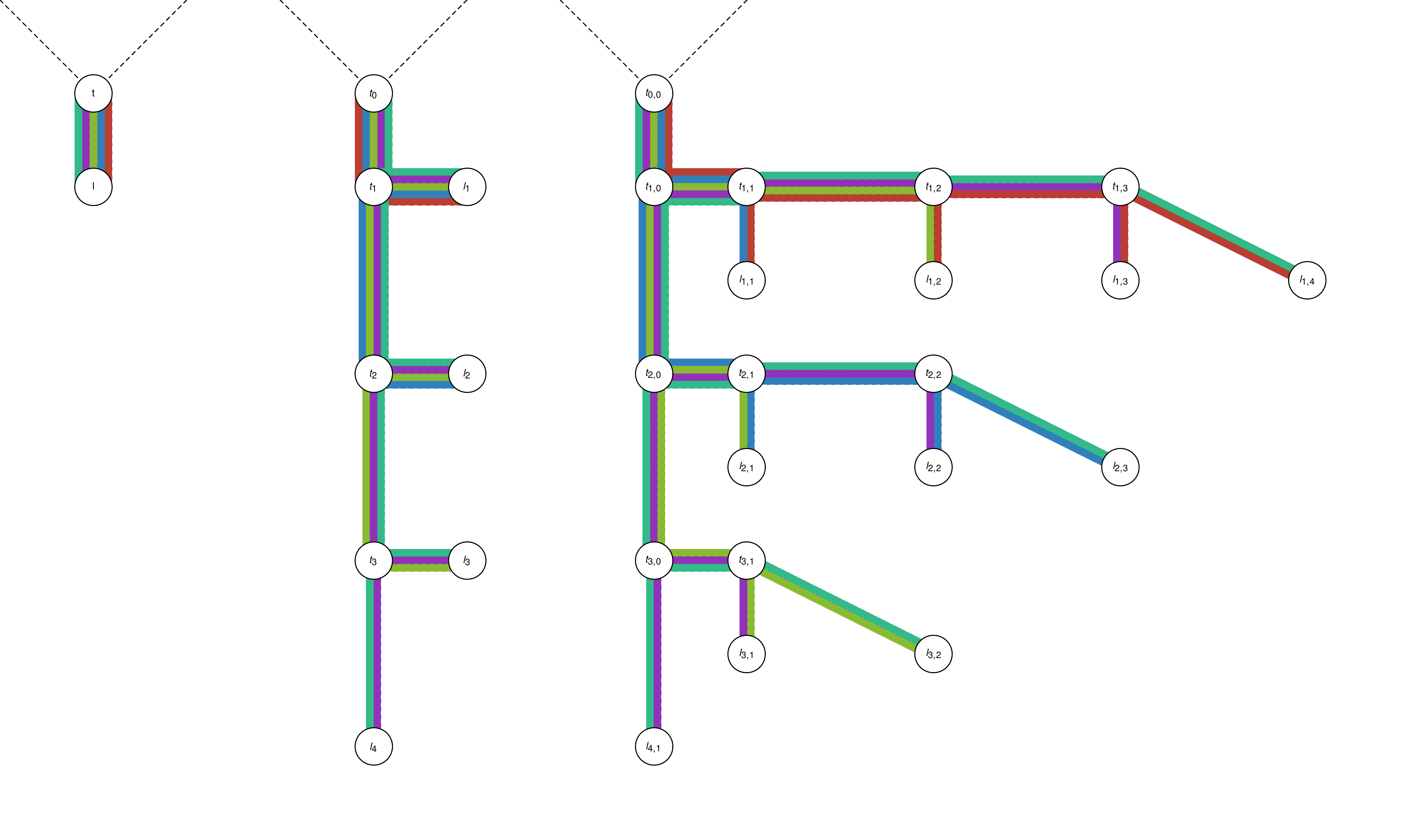}
\caption{From a tree embedding of $G$ to a branch decomposition of $L(G)$. Left: a leaf $l_0$ and neighboring vertex $t_0$ of a tree embedding.
Middle: Replacement of the leaf $l_0$ with a caterpillar subtree.
Right: Replacement of each leaf with a caterpillar subtree.}\label{fig:embedding-to-decomposition}
\end{figure}

\begin{figure}[t!]
\centering
\includegraphics[width=0.9\textwidth]{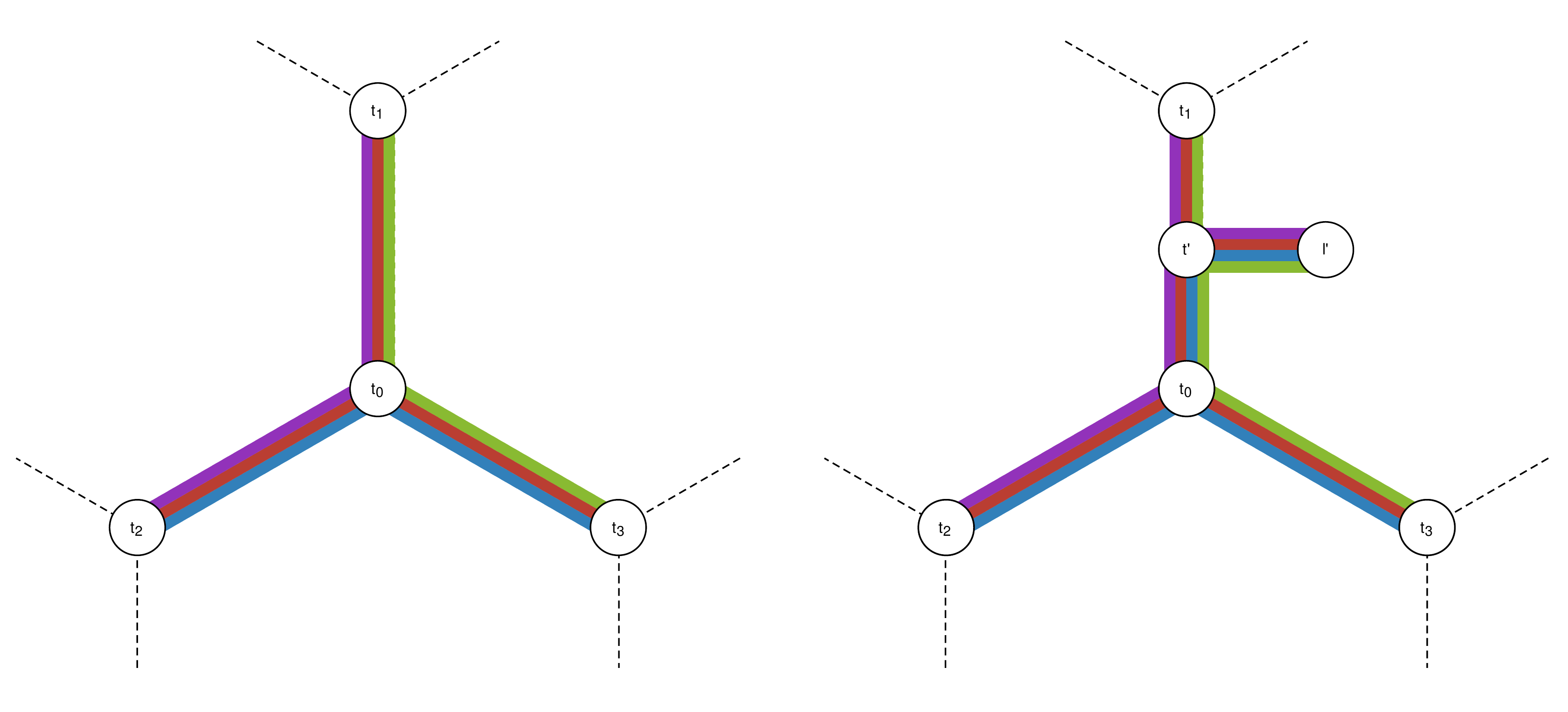}
\caption{
From a branch decomposition of $L(G)$ to a tree embedding of $G$.
Left: Part of a branch decomposition.
Right: Modified part to form a tree embedding.
}\label{fig:decomposition-to-embedding}
\end{figure}

\begin{proof}[Proof of Theorem~\ref{thm:branchwidth-equals-edge-congestion}]
First, we show how to compute a branch decomposition of $L(G)$ with width $k$ given a tree embedding of $G$ with congestion $k$, implying $\branchwidth*(L(G)) \leq \edgecongestion*(G)$.
Suppose we have a tree embedding $(T, b)$ of $G$ with edge congestion $k$.
Let $T'$ be a copy of $T$ and $b': E(L(G)) \to V(T)$ be a mapping from edges of the line graph to vertices of $T$.
In particular, for an edge $e = \{\{u, v\}, \{v, w\}\}$ of $L(G)$ set 
$b': \{\{u, v\}, \{v, w\}\} \mapsto b(v)$, i.e., $b'$ maps adjacent pairs of edges of $G$ to the same leaf mapped to from their common vertex by $b$.
Interpreted as a branch decomposition of $L(G)$, $(T',b')$ has width $k$, except that $b'$ is not injective.
We will now introduce a series of modifications to $(T', b')$ that will turn it into a proper branch decomposition with width $k$.
Note that $b'(e) = b'(f)$ if and only if $e$ and $f$ correspond to the same vertex of $G$.
For each vertex $v$ of $G$, we will replace the corresponding leaf of $T$ with a subtree whose leaves are one-to-one with the edges of $E(L(G))$ corresponding to the vertex $v$.
Consider a particular vertex $v$.
Let $l_0$ be the corresponding leaf of $T$ and $t_0$ its neighbor.
Let $(e_1, e_2, \ldots, e_{\degree_G(v)})$ be an arbitrary ordering of the edges adjacent to $v$ in $G$.
First, we replace the leaf $l_0$ with a subcubic caterpillar graph with internal vertices $(t_1, t_2, \ldots, t_{\degree_G(v) - 2})$ and leaves $(l_1, l_2, \ldots, l_{\degree_G(v) - 1})$ such that $t_i$ is adjacent to $t_{i-1}$ and $l_i$ for $i \in [\degree_G(v)-2]$ and $t_{\degree_G(v) - 2}$ is adjacent to $l_{\degree_G(v) - 1}$.
Then we set $b': \{e_i, e_j\} \mapsto l_{\min(i, j)}$.

At this point 
$\left|b'^{-1}(l_i)\right| = 
\left|\left\{\{e_i, e_j\} \middle| j > i\right\}\right|
= \degree_G(v) - i$.
For each $l_i$ we do the following.
Relabel its neighbor $t_i$ as $t_{i, 0}$.
Replace $l_i$ with another subcubic caterpillar graph with internal vertices 
$\left(t_{i,1}, t_{i, 2}, t_{i, \degree_G(v) -i - 2}\right)$ 
and leaves 
$\left(l_{i,1}, l_{i, 2}, l_{i, \degree_G(v) -i - 1}\right)$ 
such that $t_{i, j}$ is adjacent to $l_{i,j}$ and $t_{i, j-1}$ for $j \in [\degree_G(v) - 2]$ and $t_{i, \degree_G(v) - 2}$ is adjacent to $l_{i, \degree_G(v) - 1}$.
Then set 
\begin{equation}
b': \{e_i, e_j\} \mapsto 
\begin{cases}
l_{i, j - i}, & i < j, \\
l_{j, i - j}, & i > j.
\end{cases}
\end{equation}
At this point, $(T',b')$ is a proper branch decomposition of the line graph $L(G)$.
What is its width?
Let $S'_e$ be the subtree connecting $b'(\{e, f\})$ for all neighbors $f$ of $e$ in $L(G)$.
In the part of $T'$ that we didn't change, this coincides with $S_e$ of the tree embedding $(T, b)$.
The number of subtrees including the edge $\{t_0, t_{1,0}\}$ of $T'$ is the same as that including the edge $\{t_0, l_0\}$ of $T$, which is at most the edge congestion of $(T, b)$.
In particular, it is exactly $\deg_G(v)$.
These are the only subtrees that contain any part of the new parts of the tree $T'$ that we created.
The contstruction is shown for a degree $5$ vertex in Figure~\ref{fig:embedding-to-decomposition}.

Now, we show how to compute a tree embedding with congestion $k$ from a width-$k$ branch decomposition the line graph, implying $\edgecongestion*(G) \leq \branchwidth*(G)$.
Suppose we have a width-$k$ branch decomposition $(T, b)$ of $L(G)$.
Let $T'$ be a tree and $b'$ a function from $V(G)$ to $V(T')$.
Initially we set $(T', b') = (T, b)$ and iteratively modify it into a tree embedding.
For each vertex $v \in V(G)$, the neighboring edges $E_v \subset E(G) = V(L(G))$ form a clique of size $\degree_G(v)$.
Therefore, there must be some vertex $t_0$ of $T$ such that $S_e$ contains $t_0$ for all $e \in E_v$.
Let $t_1, t_2, t_3$ be the three neighbors of $t_0$ and partition $E_v$ into four (potentially empty) parts: $E_0$ contains those edges $e$ such that $S_e$ contains all of $t_0, t_1, t_2, t_3$ and $E_i$ contains those edges $e$ such that $S_e$ does not contain $t_i$, for $i = 1, 2, 3$.
Without loss of generality, assume $w(E_1) \leq w(E_2) \leq w(E_3)$.
Note that 
$\degree(v) =  
\sum_{i=0}^4 w(E_i) \geq  
\sum_{i=1}^4 w(E_i) \geq 3 w(E_1)$.
Now, subdivide the edge between $t_0$ and $t_1$, introducing a new vertex $t'$, and add a new leaf $l'$ adjacent thereto.
For all $e \in E_v$, set $b'(e) = l'$; this leaf will correspond to vertex $v$ in the tree embedding.
Note that the congestion of the edge between $l'$ and $t'$ is $\degree(v)$, and that the congestion of the edge between $t'$ and $t_0$ is $w(E_1) \leq \degree(v) / 3$ more than the congestion of the edge $\{t_0, t_1\}$ that it replaced.
If we do this for every vertex, we get a tree embedding whose congestion is at most $\degree(G) / 3$ more than the width of the branch decomposition we started with.
This is illustrated in Figure~\ref{fig:decomposition-to-embedding}.
\end{proof}

It cannot be the case that for every graph $G$, $\branchwidth*(L(G)) = \edgecongestion*(G)$.
Consider, for example, the star graph $S_k$.
Its edge congestion is at least its maximum degree $k$, but its line graph is the complete graph, whose branchwidth is $\left\lceil \frac{2k}{3}\right\rceil$. 

Consider an alternative, what we'll call the \emph{line hypergraph}, denoted $L^*(G)$, with a vertex for each edge of $E(G)$ and a \emph{hyperedge} for each vertex of $V(G)$ (rather than a clique as in the usual line graph). 
Then it is trivially true that $\branchwidth*(L^*(G)) = \edgecongestion*(G)$.
 
\bibliography{tensors-and-trees}

\end{document}